\documentclass{article}
\usepackage[utf8]{inputenc}
\usepackage{xcolor} %for colored text
\usepackage{amsmath}
\usepackage{amssymb}
\newcommand\say[1]{`#1'} %for quoting british style
\usepackage{booktabs} %for tables
\usepackage{graphicx} %for charts
\usepackage[labelfont=bf, textfont={small}]{caption}
\usepackage[scientific-notation=true, round-precision=2,round-mode=figures]{siunitx} %for conversion to scientific notation
\usepackage{geometry} % recommended on https://en.wikibooks.org/wiki/LaTeX/Page_Layout
\usepackage{setspace} % recommended on https://en.wikibooks.org/wiki/LaTeX/Text_Formatting#Line_Spacing
\definecolor{CiteGreen}{RGB}{0,160,0} %for links in citations
\usepackage{xr-hyper}
\usepackage[colorlinks, pagebackref=true, linkcolor = blue, citecolor = black, filecolor = black, urlcolor = black]{hyperref}
\usepackage[authoryear]{natbib}
\usepackage{url}

\usepackage{amsthm} %for proof
\newtheorem{theorem}{Theorem}
\newtheorem*{theorem*}{Theorem}
\usepackage{tikz}
\definecolor{darkgreen}{RGB}{0, 100, 0}
\definecolor{midnightblue}{RGB}{124, 185, 232}
\usetikzlibrary{shapes.geometric, arrows,chains}
\usetikzlibrary{shapes.multipart} %for the curly bracket label to have two lines
\tikzset{
  startstop/.style={
    rectangle, 
    rounded corners,
    minimum width=1.9cm, 
    minimum height=1cm,
    align=center, 
    draw=black, 
    fill=midnightblue!90
    },
    arrow/.style={thick,->,>=stealth},
    }
\title{A neural network model for solvency calculations in life insurance}
\author{Lucio Fernandez-Arjona \\
University of Zurich\\
}
\date{December 2019}

\begin{document}

\maketitle

\begin{abstract}
Insurance companies make extensive use of Monte Carlo simulations in their capital and solvency models. To overcome the computational problems associated with Monte Carlo simulations, most large life insurance companies use proxy models such as replicating portfolios. 

In this paper, we present an example based on a variable annuity guarantee, showing the main challenges faced by practitioners in the construction of replicating portfolios: the feature engineering step and subsequent basis function selection problem. 

We describe how neural networks can be used as a proxy model and how to apply risk-neutral pricing on a neural network to integrate such a model into a market risk framework. The proposed model naturally solves the feature engineering and feature selection problems of replicating portfolios. 
\\

\end{abstract}

\vspace*{\fill}
  {\footnotesize Correspondence to: L.Fernandez-Arjona, University of Zurich, \texttt{lucio.fernandez.arjona@business.uzh.ch}}
  
\section{Introduction}

Insurance companies rely on financial models for quantitative risk management. These risk models should be accurate and fast in terms of the calculation of risk figures such that the rapid pace of market environments is matched. Life insurance companies face the challenge of having to quickly revalue their liabilities under economic stress scenarios based on market-consistent valuation principles.

Typically, insurance liabilities exhibit features, such as options and guarantees, comparable to standard financial products. Unlike for the latter, there are generally
no closed-form formulas for the valuation of the former. Because of this, the use of numerical methods, such as Monte Carlo techniques, becomes inevitable. However, the choice of the specific technique to be used is the key factor in the accuracy and speed of the calculation of risk figures.

Nested Monte Carlo, a straightforward approach to this problem, is computationally burdensome to the point of being infeasible in most cases. Other standard techniques to approach this problem are the least squares Monte Carlo approach (LSMC) and the replicating portfolio approach (RP). These techniques, and their advantages and disadvantages, will be described in more detail in the following sections.

The availability of accurate and fast valuation methods is of great interest to risk management practitioners in life insurance companies. In the last decade, machine learning models based on neural networks have emerged as the most accurate in many fields, among them image recognition, natural language understanding and robotics.
While not the first to apply neural networks to life insurance modelling, this paper incorporates the idea of risk-neutral valuation of neural networks to achieve higher accuracy than other existing models while remaining within a realistic computational budget.

A review of the literature shows that several papers work with machine learning techniques to address the problem of calculating the value and risk metrics of complex life insurance products. 

One family of methods starts with exact information about a small number of contracts, and then uses spatial interpolation to generate valuations for all contracts. For example, \cite{gan2015valuation} use a clustering technique for the selection of the representative contracts and then apply a functional data analysis technique, universal kriging, for the interpolation. \cite{hejazi2016neural} propose using a neural network model for the interpolation step instead of traditional interpolation techniques. These methods interpolate among policies for a fixed set of economic scenarios, hence reducing the computational burden of calculating the value of the insurance contract. They address the valuation problem but on their own they are not enough to solve the computational problem of calculating risk metrics.

The effective calculation of risk metrics is addressed by a family of methods that take portfolio valuations (or cash flows) as inputs and use regression methods to construct a model capable of producing the real-world distribution of values of the insurance portfolio. These methods can be classified in two groups, \say{regress-now} and \say{regress-later}, a classification first proposed by \cite{glasserman2002simulation}. Regress-later methods perform better than regress-now methods, as shown by \cite{beutner2013fast}. Examples of both types of methods in this family can be found in \cite{beutner2016theory} and \cite{castellani2018investigation}. The former presents a regress-later model based on an orthogonal basis of piece-wise linear functions, and the latter a regress-now model based on neural networks. LSMC and replicating portfolios belong to this overarching family of methods, and any machine learning approach in this family can be a direct replacement for them.

In this paper we present a regress-later model based on neural networks, including a closed-form formula for its risk-neutral valuation, and compare it to a benchmark implementation of replicating portfolios. We focus on this comparison in order to answer a question of relevance to practitioners: \say{Can neural networks provide better results than existing methods used in the industry?} A formal mathematical treatment of the methods involved can be found in the already mentioned \cite{beutner2013fast} (regress-later vs regress-now methods), \cite{natolski2014mathematical}, and \cite{cambou2018replicating} (RPs) and \cite{bauer2010solvency} (LSMC).

We have not found in the literature previous work on comparing replicating portfolios (that is, regress-later method with financial instruments as basis functions) to other approximation techniques. There are, however, several papers that compare LSMC approaches: \cite{bauer2010solvency} present a comparison between LSMC and nested Monte Carlo, and \cite{pelsser2016difference} present a comparison between LSMC regress-now and LSMC regress-later. 

Against the background described above, this paper's contribution is threefold:
\begin{itemize}

    \item it presents a reproducible replicating portfolio approach suitable for benchmarking in a research context,
    
    \item it introduces a risk-neutral valuation formula for a class of neural networks with multivariate normally distributed inputs (such as discrete time Brownian motion processes),
    
    \item it builds a regress-later methodology based on neural networks and compares the quality of the economic capital calculations between this method and the replicating portfolio approach.
\end{itemize}

The neural network model improves on those in use in the industry and some of those presented in the literature. In comparison with the (regress-now) neural network model in \cite{castellani2018investigation}, our model is based on a regress-later approach, which---as described in the literature---provides more accurate results than regress-now models. In comparison with the regress-later approach in \cite{beutner2016theory}, the neural network approach allows us to avoid having to define the arbitrary $A_T(Z)$ (dimensionality reduction function) and the hypercube grid. Under a neural network model, those parameters are data driven and determined as part of the optimization. 

The rest of the paper is structured as follows: Section \ref{sec:capital calculations} describes the solvency capital calculation problem in detail, together with possible solutions based on nested Monte Carlo and proxy models. Section \ref{sec:math} summarizes the mathematical framework common to all regression models (both regress-now and regress-later), Section \ref{sec:nn} describes the proposed neural network approach, and presents the risk-neutral valuation of a class of neural networks. Section \ref{sec:qualitative} discusses important qualitative aspects of the models, such as model complexity and feature engineering. Finally, Section \ref{sec:experiments} describes the numerical experiments and presents the results.

\section{Solvency capital calculation problem in life insurance} \label{sec:capital calculations}
Solvency regimes---including Solvency II or Swiss Solvency Test---require companies to hold capital in excess of a legal minimum that is based on the amount necessary to remain solvent with high confidence in a one-year period.

The above implies the determination of the distribution of the value of the asset--liability portfolio at the end of the one-year period. We call the value of the asset--liability portfolio $V_t$, and therefore $V_1$ is the value at the end of the first year. The solvency capital requirement is then determined relative to a risk metric applied to this distribution:
\begin{itemize}

\item Value at risk (Solvency II) 
        \[
         \mathit{VaR}_{\alpha }(V_1) =-\inf {\big \{}v :F_{V_1}(v)>\alpha {\big \}}=F_{-V_1}^{-1}(1-\alpha ) .
        \]
\item Expected shortfall (Swiss Solvency Test)
\[
         \mathit{ES}_{\alpha }(V_1) = -{\frac {1}{\alpha }}\int _{0}^{\alpha }{\mbox{VaR}}_{\gamma }(V_1)d\gamma.
  \]      
\end{itemize}

If $F_{V}^{-1}$ is not known (and this is usually the case), then we must simulate $\{V_1^{(i)}\}_{i=1:M}$ and then calculate the risk metric on the empirical (simulated) distribution. These $M$ samples are referred to as real-world (or \say{natural}) simulations. 

Having described how risk calculations usually depend on a Monte Carlo sampling of $V_1$, $\{V_1^{(i)}\}_{i=1:M}$, we now focus on the calculation of each individual $V_1^{(i)}$. These represent the value of the asset--liability portfolio at the end of the one-year period. The valuation must be carried out on a market-consistent basis, which means applying risk-neutral valuation:
    \begin{equation} \label{eq:riskneutral}
    V_t = E_t^{\mathbb{Q}}\big[\sum_{\tau > t} \mathit{CF}_{\tau}(X)\big].
    \end{equation}
    
In the formula above the value of the portfolio is the expectation over the risk-neutral measure $\mathbb{Q}$ of future discounted cash flows, $CF_{\tau}(X)$. These cash flows depend on a set of economic variables $X$. In turn, $X$ depends on a smaller set of normal random drivers $\xi$---that is, $X=X(\xi)$. This implies that there is a $CF'_{\tau}(\xi)$ such that $CF' = CF \circ X$.

When calculating $V_1$ for solvency capital purposes, the real-world simulations determine an empirical distribution of $X_{0:1}$ ($X$ between 0 and 1), and for each sample in such distribution there is a risk-neutral distribution of $X$ after $t=1$---which we call $X_{1:T}$---over which $V_1$ must be calculated.

For simple products the calculation of $V_1$ is straightforward. However, products that lack a closed form formula, $\{V_1^{(i)}\}_{i=1:M}$ must be approximated by some $\{\widehat{V}_1^{(i)}\}_{i=1:M}$. Most complex products fall under this case.

\subsection{Nested Monte Carlo}
A simple solution for approximating $V_1$ is to apply a Monte Carlo approach:
    \[ 
     \widehat{V}_t = V_{\mathit{MC}} = \frac{1}{N} \sum_{j=1}^{N} \sum_{\tau > t} \mathit{CF}_{\tau}(X_{t:T}^{(j)} | X_{0:t}).
    \]
    
This formula implies taking $N$ samples and calculating the average of the discounted simulated cash flows. Since this must be done for each real-world simulation $i$ to be able to construct $\{V_1^{(i)}\}_{i=1:M}$, the full simulated distribution is given by
\[
\{\widehat{V}_t^{(i)}\}_{i=1:M} = \Big\{\frac{1}{N} \sum_{j=1}^{N} \sum_{\tau > t} \mathit{CF}_{\tau}(X_{t:T}^{(j)} | X^{(i)}_{0:t})\Big\}_{i=1:M}.
\]

The set of risk-neutral scenarios are called the inner scenarios because they are constructed for each outer scenario $i$ to estimate the value of the portfolio conditional on the information at time 1. Since a total of $M \times N$ simulations are required, a nested Monte Carlo approach is usually infeasible. Most insurers, therefore, use other approximation methods for $\widehat{V}_1$, the most popular being  \say{least squares Monte Carlo} (LSMC) and \say{replicating portfolios} (RPs). Both these methods are based on a regression approach. For an analysis and comparison of nested Monte Carlo to regression-based methods, the reader is referred to \cite{broadie2015risk}.

\subsection{Alternatives to nested Monte Carlo}
Given the computational difficulties of nested Monte Carlo, many methods have been proposed in the literature, some of which are in place in the industry. With regard to those in place, it is important to note that none of these proxy models replace the original, full insurance cash flow model of the asset--liability portfolio. These methods focus on allowing the solvency capital to be calculated with fewer executions of that model.

LSMC, originally introduced for pricing American style derivatives by \cite{LongstaffSchwartz}, is used to reduce the number of necessary risk-neutral simulations by finding a polynomial approximation of the portfolio value as a
function of the risk drivers. The coefficients of the polynomial expansion are obtained from a regression against a reduced-size nested Monte Carlo. LSMC is usually applied in the industry as, in the terminology of \cite{glasserman2002simulation}, a \say{regress-now} approach, as described in \cite{bauer2010solvency}. However, polynomial approximations do not need to be restricted to \say{regress-now} applications. 

The replicating portfolio approach is based on running a regression against the portfolio cash flows, but instead of polynomials the model uses a set of financial securities as basis functions. The problem is formulated as a linear optimization problem ($L_1$ or $L_2$) where the objective is to minimize the differences between the cash flows of the asset--liability portfolio and the replicating portfolio. The output is a portfolio of financial instruments that reproduces the payout of the life insurance portfolio as closely as possible. For reference, \cite{natolski2014mathematical} analyse in detail some of the popular approaches to constructing replicating portfolios. \cite{vidal2009replication} and \cite{chen2012cashflow} look at replicating portfolios from a more practical point of view and \cite{fernandez2016large} provide a description of a real-world implementation in the insurance industry.

In this paper we present a method based on neural networks that combines the strengths of LSMC and replicating portfolios while providing higher accuracy in a setting with a realistic amount of inputs and computing power. We stress this last point since it is common for studies on neural networks to provide results based on millions of input samples, which would not be practical in the real world.

\section{Mathematical formulation} \label{sec:math}
Both regress-now and regress-later models are estimators for the value of the asset--liability portfolio, $\widehat{V}_t$. 

Regress-now models approximate the value function by estimating the coefficients $\{w_k\}$ in 
\[ \widehat{V}_{t}^{(i)}(X) = \sum_k w_k \phi_k(X_{0:t}^{(i)}).\]

This estimation is done via regression, minimizing the squared error
\[\sum_{i=1}^{m} \Big(\sum_k w_k \phi_k(X_{0:t}^{(i)}) - \tilde{V}_{MC}^{(i)}\Big)^2, \]
where $\tilde{V}_{MC}^{(i)}$ differs from $V_{MC}^{(i)}$ in that it is calculated with a very low number $n$ of inner simulations, instead of using $N$ simulations as in nested Monte Carlo. It is also worth noting that the regression is performed over $m<<M$ outer scenarios. LSMC uses a polynomial basis for $\{\phi_k\}$.
Once the model has been calibrated, it can be used to \say{predict} (in its machine learning sense) all $M$ scenarios, which were not part of the training data.

Regress-later models approximate the value function by estimating the coefficients $\{w_k\}$ necessary to build the following estimator: 

\begin{equation} \begin{split} \label{eq:regresslater}
    \widehat{V}_t^{(i)} &=  E_t^{\mathbb{Q}}\big[\sum_{\tau > t} \widehat{CF}_{\tau}(X)\big] \\ 
    &= E_t^{\mathbb{Q}}\big[\sum_{\tau > t} \sum_k w_{k,\tau} \phi_{k,\tau}(X_{0:\tau})\big] \\ 
    &=\sum_{\tau > t} \sum_k w_{k,\tau} E_t^{\mathbb{Q}}\big[ \phi_{k , \tau}(X^{(i)})\big].
\end{split} \end{equation}

The equation above shows that, while more accurate than regress-now models, regress-later models impose an additional requirement: to be able to calculate $E_t^{\mathbb{Q}}[ \phi_{k , \tau}]$. In the best case, there is a closed-form formula. In the worst case, it can be done via Monte Carlo, but the computation cost will partially or completely offset the computation gains of avoiding nested Monte Carlo on the full model.

This estimation is done via regression, minimizing the error

%\[\sum_{i=1}^{m} \sum_{j=1}^{n} \sum_{\tau > t} \Big( \sum_k w_k \phi_k(X_{t:T}^{(j)} | %X^{(i)}_{0:t})) - \mathit{CF}_{\tau}(X_{t:T}^{(j)} | X^{(i)}_{0:t}))^2 \]
%or absolute error
\[\sum_{i=1}^{m} \sum_{j=1}^{n} \sum_{\tau > t} \Big| \sum_k w_{k,\tau} \phi_{k,\tau}(X_{t:T}^{(j)} | X^{(i)}_{0:t})) - \mathit{CF}_{\tau}(X_{t:T}^{(j)} | X^{(i)}_{0:t})\Big|^p, \]
where (as in the LSMC case) $m<<M$ but (unlike in that case) $n$ could be as large as $N$ if required.

Typically this regression is done minimizing squared errors ($p=2$) but some companies use absolute errors ($p=1$). The approximation is based on a linear regression at each $\tau$ of $CF_{\tau}(\cdot)$ against $\{\phi_{k, \tau}(\cdot)\}$, the cash functions of a set of financial instruments (bonds, swaps, equity options). This is the usual case in the industry, although some older implementations used grouped cash flows (in time buckets) and some papers, like \cite{cambou2018replicating}, present the topic in terms of terminal values (that is, all time steps grouped).

The neural network approach proposed in this paper will take the form of a regress-later method, but using neural network structure for basis functions and performing a non-linear optimization to find its parameters.

\section{A neural network approach} \label{sec:nn}
Artificial neural networks, more commonly referred to as neural networks, constitute a broad class of models. Among them, the simplest is the single-layer perceptron, which we use for our model. 

The single-layer perceptron, a type of feed-forward network, is a collection of connected units or nodes arranged in layers. Nodes of one layer connect only to nodes of the immediately preceding and immediately following layers. They are fully connected, with every node in one layer connecting to every node in the next layer. The layer that receives external data is the input layer. The layer that produces the ultimate result is the output layer. In between there is one hidden layer. Each node is a simple non-linear function $\phi(\cdot)$ applied to a linear combination of its inputs---that is, those nodes to which is connected. An example of this architecture is shown in Figure \ref{fig:nn_structure}, for three inputs and a hidden layer with a width of four nodes.

The single-layer perceptron is a universal function approximator, as proven by the universal approximation theorem (\cite{hornik1991approximation}).

\begin{figure}
    \begin{center}
    \includegraphics[trim={4cm 11cm 4cm 10.5cm},clip]{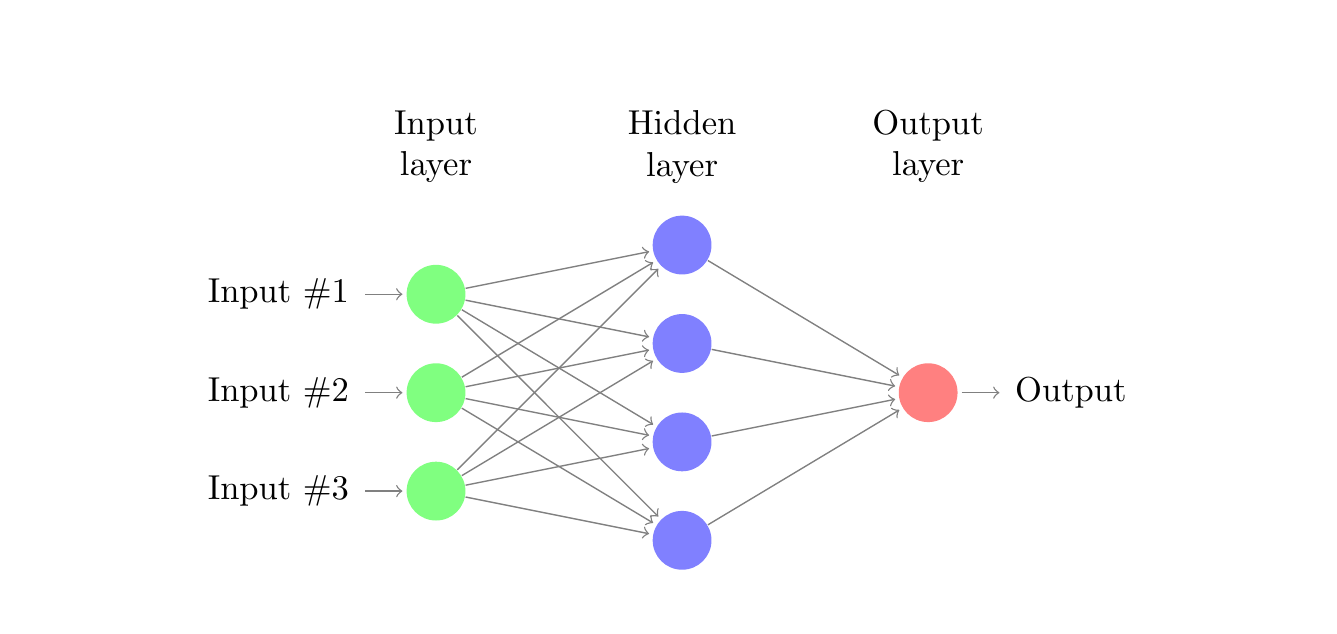}
    \end{center}
    \caption{Single-layer perceptron structure}
    \label{fig:nn_structure}
\end{figure}

The first neural network model that we present is a direct equivalent of Equation \eqref{eq:regresslater} for the case of single-layer perceptron:
\[\widehat{V}_t = E_t^{\mathbb{Q}}\big[\sum_{\tau > t} \sum_k w_{k,\tau} \phi_{k,\tau}(X_{0:\tau})\big]= E_t^{\mathbb{Q}}\big[\sum_{\tau > t} \sum_k w_{k,\tau} \phi(\mathbf{v}^\intercal_{k,\tau}  X_{0:\tau})\big], \]
where $\phi$ is the activation function, $w_{k,\tau}$ the weights of the liner (output) layer, and $\mathbf{v}_{k,\tau}$ the weights of the hidden layer. Since $X_0$ is a constant (it describes the initial conditions of the simulation) there is no need for an explicit bias term since the first component of $\mathbf{v}$ will be the constant term.

In this first neural network model, we do not know the distribution of $X$, which contains arbitrary economic variables, and therefore the risk-neutral expectation cannot be calculated in closed form. When this expectation is required, as in Section \ref{sec:results}, we present results for this model based on Monte Carlo valuation. Since this model's input is $X$---the vector of economic variables---we will refer to this model as the \say{nn econ} when showing results.

The second neural network model is more interesting because it allows a closed form valuation. When discussing Equation \eqref{eq:riskneutral}, we pointed out that $X$ is modelled as a function of a smaller set of normal random drivers $\xi$---that is, $X=X(\xi)$ and $\dim(\xi) < \dim(X)$. Using this fact, we express our second model as 
\[\widehat{V}_t = E_t^{\mathbb{Q}}\big[\sum_{\tau > t} \sum_k w_{k,\tau} \phi(\mathbf{v}^\intercal_{k,\tau}  \xi_{0:\tau})\big]; \]
that is, we use $\xi$ as input instead of $X$ (for symmetry, we keep a $\xi_o$ component, which is not random but constant in order to provide a bias term). 

Using $\xi$ as input brings two advantages and one disadvantage. On the positive side, the dimensionality of $\xi$ is smaller than that of $X$ ($dim(\xi) < dim(X)$) and its components are uncorrelated with each other ($\sigma(\xi_i, \xi_j) = \delta_{ij}$ but $\sigma(X_i, X_j) \neq \delta_{ij}$). This makes solving the non-linear optimization problem much easier, requiring fewer samples and shorter training time to converge. Most importantly, the normal distribution of $\xi$ allows a closed-form solution to the risk-neutral expectation as described in Section \ref{sec:nnclosedform}. On the negative side, $CF'(\xi)=CF \circ X(\xi)$ is a more complex function than $CF(X)$ so---all else being equal---we would expect to need a bigger neural network, more samples, and a longer training time. Given these advantages and disadvantages, it is not possible to tell analytically which model will show higher accuracy. Both will, therefore, be tested.
 Since this model's input is $\xi$---the vector of normal random variables---we will refer to this model as the \say{nn rand} when presenting results.

\subsection{The risk-neutral value of a neural network} \label{sec:nnclosedform}
We have claimed that the second neural network model has a closed-form solution to the risk-neutral expectation. In this section we will show its derivation, which, as far as we know, has not appeared before in the literature.

\begin{theorem}\label{th:nn}
For a normally distributed $\xi$, the time-t risk-neutral expectation,
\[ E_t^{\mathbb{Q}}\big[w_{0,\tau}+\sum_{k=1}^K w_{k,\tau} \phi(\mathbf{v}^\intercal_{k,\tau}  \xi_{0:\tau})\big], \]
of a single-layer network with $K$ hidden nodes and a ReLu activation function $\phi$ that models the cash-flows at time $\tau$ is given by

\begin{equation}\label{eq:nn_expectation}
     w_{0,\tau}+ \sum_{k=1}^K w_{k,\tau}  \frac{1}{2}\left[{_t}{\mu_{k,\tau}} + {_t}\sigma_{k,\tau}\sqrt{\frac{2}{\pi}}\exp\left({-\frac{{_t}\mu_{k,\tau}^2}{2{_t}\sigma_{k,\tau}^2}}\right)+ {_t}\mu_{k,\tau}\left(1-2\Phi\left(-\frac{{_t}\mu_{k,\tau}}{{_t}\sigma_{k,\tau}}\right)\right)\right],
\end{equation}
       \begin{equation*}\begin{split}
       {_t}{\mu_{k,\tau}} &=\sum_{i=0}^{\min(t, \tau)} \mathbf{v}^\intercal_{i, k,\tau}  \xi_{i}, \\
       {_t}\sigma_{k,\tau}^{2} &= \sum_{i=t+1}^{\tau} \lVert\mathbf{v}^\intercal_{i, k,\tau}\lVert^2 .
    \end{split} \end{equation*}
\end{theorem}

\begin{proof}
Starting from the full network
\begin{equation} 
     E_t^{\mathbb{Q}}\big[w_{0,\tau}+ \sum_{k=1}^K w_{k,\tau} \phi(\mathbf{v}^\intercal_{k,\tau}  \xi_{0:\tau})\big]
   = w_{0,\tau}+ \sum_{k=1}^K w_{k,\tau} E_t^{\mathbb{Q}}\big[\phi(\mathbf{v}^\intercal_{k,\tau}\xi_{0:\tau})\big],
\end{equation}
and then focusing on the expectation of a single node, we obtain
\begin{equation} \begin{split}
E_t^{\mathbb{Q}}\big[\phi(\mathbf{v}^\intercal_{k,\tau} \xi_{0:\tau})\big]&=
E_t^{\mathbb{Q}}\big[\max(\mathbf{v}^\intercal_{k,\tau} \xi_{0:\tau}, 0) \big] \\
&= E_t^{\mathbb{Q}}\big[\frac{\mathbf{v}^\intercal_{k,\tau} \xi_{0:\tau} + |\mathbf{v}^\intercal_{k,\tau}  \xi_{0:\tau}|}{2} \big] \\
&=\frac{1}{2}\Big[ E_t^{\mathbb{Q}}\big[\mathbf{v}^\intercal_{k,\tau}\xi_{0:\tau}\big] +E_t^{\mathbb{Q}}\big[ |\mathbf{v}^\intercal_{k,\tau} \xi_{0:\tau} | \big]\Big].
\end{split} \end{equation}

Since $\mathbf{v}^\intercal_{k,\tau}\xi_{0:\tau}$ is normally distributed, it is defined by its mean and standard deviation, $\mu_{k,\tau}$ and $\sigma_{k,\tau}$. Its conditional expectation at time $t$ is 
\[
E_t^{\mathbb{Q}}\big[\mathbf{v}^\intercal_{k,\tau} \xi_{0:\tau}\big]=  {_t}{\mu_{k,\tau}} =\sum_{i=0}^{\min(t, \tau)} \mathbf{v}^\intercal_{i, k,\tau}  \xi_{i},
\]

and its conditional variance at time t is
\[{_t}\sigma_{k,\tau}^{2} = \sum_{i=t+1}^{\tau} \lVert\mathbf{v}^\intercal_{i, k,\tau}\lVert^2 .
\]
If $\tau \leq t$, then its conditional variance is 0.

Since $\mathbf{v}^\intercal_{k,\tau}\xi_{0:\tau}$ is normally distributed, $|\mathbf{v}^\intercal_{k,\tau}  \xi_{0:\tau}|$ follows a folded normal distribution, with conditional expectation at time $t$

\[
E_t^{\mathbb{Q}}\big[|\mathbf{v}^\intercal_{k,\tau} \xi_{0:\tau}|\big]=  {_t}\sigma_{k,\tau}\sqrt{\frac{2}{\pi}}\exp({-\frac{{_t}\mu_{k,\tau}^2}{2{_t}\sigma_{k,\tau}^2}})+ {_t}\mu_{k,\tau}\big(1-2\Phi(-\frac{{_t}\mu_{k,\tau}}{{_t}\sigma_{k,\tau}})\big).\]
Therefore, the expectation of each hidden node is

\[E_t^{\mathbb{Q}}\big[\phi(\mathbf{v}^\intercal_{k,\tau} \xi_{0:\tau})\big]=
\frac{1}{2}\left[{_t}{\mu_{k,\tau}} + {_t}\sigma_{k,\tau}\sqrt{\frac{2}{\pi}}\exp\left({-\frac{{_t}\mu_{k,\tau}^2}{2{_t}\sigma_{k,\tau}^2}}\right)+ {_t}\mu_{k,\tau}\left(1-2\Phi\left(-\frac{{_t}\mu_{k,\tau}}{{_t}\sigma_{k,\tau}}\right)\right)\right],
\]

which leads to the expectation of the full network in \eqref{eq:nn_expectation}.

\end{proof}

\section{Qualitative comparison} \label{sec:qualitative}
Besides the quantitative experiments, whose results we show in Section \ref{sec:results}, there are some qualitative differences of importance to practitioners. One of them is the complexity of the model, as measured by the number of modules and equations required to implement it. A second one is the amount of expert judgement required in the feature engineering. A third one is how to prevent overfitting of the training data.

\subsection{Model complexity}

Figure \ref{fig:model_structure} presents a comparison of the module structure. We describe below the sequence of calculations required for the training phase, the prediction phase being very similar with the exception that cash flows are replaced by prices (closed-form or Monte Carlo, depending on the model and the instruments used).

All models require a random number generator. Replicating portfolios and the first neural network model require the generation of the economic variables $X$. In the insurance industry, these two modules are usually grouped into the so-called economic scenario generator (ESG). The second neural network model does not need $X$ for training or prediction.  Finally, the replicating portfolio model requires the generation of instrument cash flows in order to produce the inputs to the optimization problem. Despite their reputation for complexity, neural networks actually lead to a simpler model, at least when measured by the number of equations and components necessary for their implementation.

\begin{figure}
    \begin{center}
    \includegraphics[trim={4cm 11cm 4cm 10.5cm},clip]{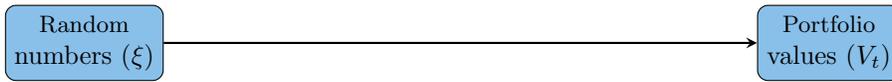}
    \end{center}
    \caption{Comparison of model structures}
    \label{fig:model_structure}
\end{figure}

%In the prediction phase, the replicating portfolios rely on closed-form formula valuations for simple instruments. When using exotic derivatives, it is possible that numerical methods are required. The first neural network model relies on a Monte Carlo valuation while the second neural network model, as described in the next section, is able to produce valuation using a closed form formula.

\subsection{Feature engineering }
Any replicating portfolio model requires a substantial degree of expert judgement in deciding which instrument cash flows to use in the model. When working with simple liabilities, the decision is not hard since the simplest instruments, such as bonds and equity forwards, will work well. As the liabilities grow in complexity, and if the most obvious derivatives (swaps, swaptions, European and Asian options) are not enough to capture the behaviour of the liabilities, the practitioner faces the extremely difficult task of figuring out which is the correct derivative to add to the existing mix. Since financial instruments as a whole do not form a structured basis of any meaningful space of functions, it is not possible to explore in a systematic way the set of all possible financial instruments until the best solution is found. Even worse, each attempt (adding of a new asset class) requires a substantial amount of implementation work before the results can be seen.
In contrast, neural networks provide certain guarantees of convergence by simply increasing the width of the hidden layer (adding more nodes). This guarantee is provided by the \say{universal representation theorem}, of which different versions exist (among them \cite{hornik1991approximation} and \cite{ hanin2017approximating}). At least for the classes of functions covered under these theorems, the search for better results is extremely straightforward: just one parameter that is a direct input in any neural network software library. No new equations need to be implemented, only one input change in the existing model in required.
Even when no new asset classes are required, the feature engineering problem in replicating portfolios still exists. Each asset class can have tens, hundreds, or thousands of individual instruments. For example, there is one zero coupon for each possible maturity. Even worse is the case of swaptions: having to choose from a combination of a set of maturities, tenors, and strikes leads to having to select thousands of instruments. All these basis elements must be completely calculated before being fed to the linear regression problem. In contrast, neural networks create their own features adaptively based on the data. This, of course, has the downside that it requires more training data than a comparable replicating portfolio model for which the expert judgement selection has been carried out correctly.

\subsection{Feature selection}
The reverse problem to not having enough financial instruments or not having asset classes that are complex enough is the problem of having too many instruments and asset classes. Feeding thousands of instruments to the regression problem runs the very real risk of overfitting the training data. This problem has not been entirely solved by practitioners in the industry, although the popularization of machine learning software libraries has allowed big improvements in recent years compared with the completely manual approach used five or ten years ago. In this paper, we solve this problem by using Lasso regression (\cite{tibshirani1996regression}) with the Bayesian information criterion to choose the regularization parameter (\cite{zou2007degrees}). For the neural network model, the constraints on model complexity are provided by selecting the layer width via cross-validation.

\section{Numerical experiments} \label{sec:experiments}
The goal of this section is to provide a quantitative comparison between a neural network model and a replicating portfolio model. It is important to clarify that there are many possible neural network models (many hyper-parameters and architectural choices) as well as many possible replicating portfolio models (many options of asset classes in the instrument universe and other architectural choices). Furthermore, the results presented here are for one particular scenario generator and one insurance product. It is not possible to generalize the conclusions drawn from these results to all models and all products. However, the choices made in this example are mainstream and robust. We would, therefore, expect the conclusions to extend to many real-world situations.

\subsection{Experimental setup}
In order to provide a comparison between replicating portfolios and another model, it is necessary to first have access to simulated cash flows of an insurance product. In the absence of open source libraries or data sets there has been no option but to build our own economic scenario generator and insurance model. Hoping that the data might help others in their research in the field, we have published the full data set in Mendeley Data (\cite{dataset}).

The high-level structure of the insurance model is described in Figure \ref{fig:insurance}. The first module is the normal random variable generator, which feeds the second module, the economic scenario generator. We use a combination of a one-factor Hull--White for the short rate and a geometric Brownian motion process for equity returns (in excess of risk-free rates). For the interest rate model, we use the formulas in \cite{glasserman2013monte}.

\begin{figure}[]

\begin{center}
    \includegraphics[trim={5cm 11cm 5cm 10.5cm},clip]{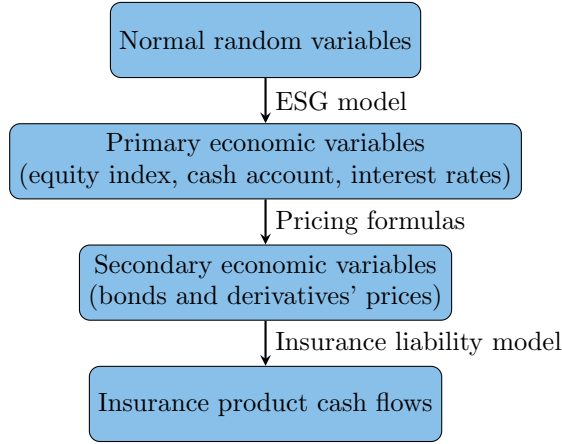}
    \end{center}
    \caption{Modular structure of the insurance model} \label{fig:insurance}
\end{figure}

In addition to the published data set, the full code of the scenario generator is available in open source form at \url{https://gitlab.com/luk-f-a/EsgLiL}. The generator is entirely written in Python. It uses NumPy (\cite{numpy}, \cite{van2011numpy}) for array operations, pandas (\cite{mckinney-proc-scipy-2010}) for data aggregation, scikit-learn (\cite{pedregosa2011scikit}) for linear regressions and neural networks training, and joblib for parallelization.

This scenario generator produces the evolution of the short rate, the cash account and the equity index. From these, we derive the rest of the asset prices: bonds, swaptions, and equity options. The last module is the insurance product, a variable annuity guarantee known as \say{guarantee return on death}. The simulation of the insurance product begins with a policyholder population of 1,000 customers of ages 30 to 70. The population evolves according to a Lee--Carter stochastic mortality model (we have used the same parameters as in \cite{leecarter}), which provides a trend and stochastic fluctuations. Each customer starts the simulation with an existing investment fund and a guaranteed level. In each time period, they pay a premium, which is used to buy assets; these assets are deposited in the fund.  The fund is rebalanced at each time period to maintain a target asset allocation.  The value of the fund is driven by the inflows from premiums and the market value changes are driven by the interest rate, equity, and real estate models. At each time step, a proportion of policy holders (as determined by the stochastic life table) die and the investment fund is paid out to the beneficiaries. If the investment fund were below the guaranteed amount, the company will additionally pay the difference between the fund value and the guaranteed amount.  The guaranteed amount is the simple sum of all the premiums paid over the life of the policy.  Over the course of the simulation the premiums paid increase the guaranteed amount for each policy.

All policies have the same maturity date, of forty years. Those policyholders alive at maturity receive the investment fund or the guaranteed value, whichever is the higher.

\subsection{Ground truth and benchmark value}
The quality comparison across methods requires establishing a \say{ground truth}---the values of the risk metrics calculated in an exact way. Given the lack of closed-form formulas this is not possible, so we settle for performing comparisons against a benchmark value calculated in the most reliable way possible. We do this by running an extremely large Monte Carlo simulation, with 100,000 outer simulations each with 10,000 inner simulations. 

The results of this calculation are shown in Table \ref{tab:groundtruth}. The centre column shows the mean present value, the risk-neutral value of the guarantee. To the left and right we see the expected shortfall and value at risk. Both are expressed as the change in monetary value in respect of the mean---that is as $V_1^{\mathit{tail}} - V_0$.

\begin{table}[h]
\centering
\small
\caption{Benchmark values} \label{tab:groundtruth}
    \begin{tabular}{@{}lrrrrr@{}}
    \toprule
    {} & \textbf{Left ES} & \textbf{Left VaR} &  \textbf{Mean} & \textbf{Right VaR} & \textbf{Right ES} \\
    \midrule
    \textbf{Large nMC} &     \num{-3888721.24} &       \num{-3317369.58} &   \num{-5474806.53} &       \num{2470920.12} &     \num{2740604.28} \\

    \bottomrule
    \end{tabular}
\end{table}

\subsection{Nested Monte Carlo and replicating portfolio benchmarks}
We use two established methods to benchmark our proposed model: nested Monte Carlo (nMC) and replicating portfolios (RPs). In all cases we set a computational budget of 10,000 training samples. In the case of nested Monte Carlo, this budget is split into 100 outer and 100 inner simulations. This split was selected for being the combination with the largest number of inner simulations possible (lowest bias) with a minimum of 100 outer simulations (to be able to calculate the 1 percent expected shortfall). For reference, note the large difference with the benchmark value calculation made with 100,000$\,\times \,$10,000 simulations.
Since the training samples are randomly drawn, each estimator is itself a random variable and we treat them as such. For each method, we calculate 100 macro-runs with new random numbers in order to obtain an empirical distribution of the estimators.

For the given sample budget of 10,000 simulations, the nested Monte Carlo estimator requires choosing the mix between outer and inner simulations. A higher number of outer simulations decreases variance, but a lower number of inner simulations increases bias. This effect is known as the bias--variance trade-off. Figure \ref{fig:nmc_results} shows this effect quite clearly using a violin plot, showing the increasing bias and decreasing variance from left to right. Each individual distribution is annotated with a percentage showing the mean absolute percentage error of the estimator. We describe this error measure in more detail in the next section. 

\begin{figure}[h]
    \centering
    \includegraphics[width=\columnwidth]{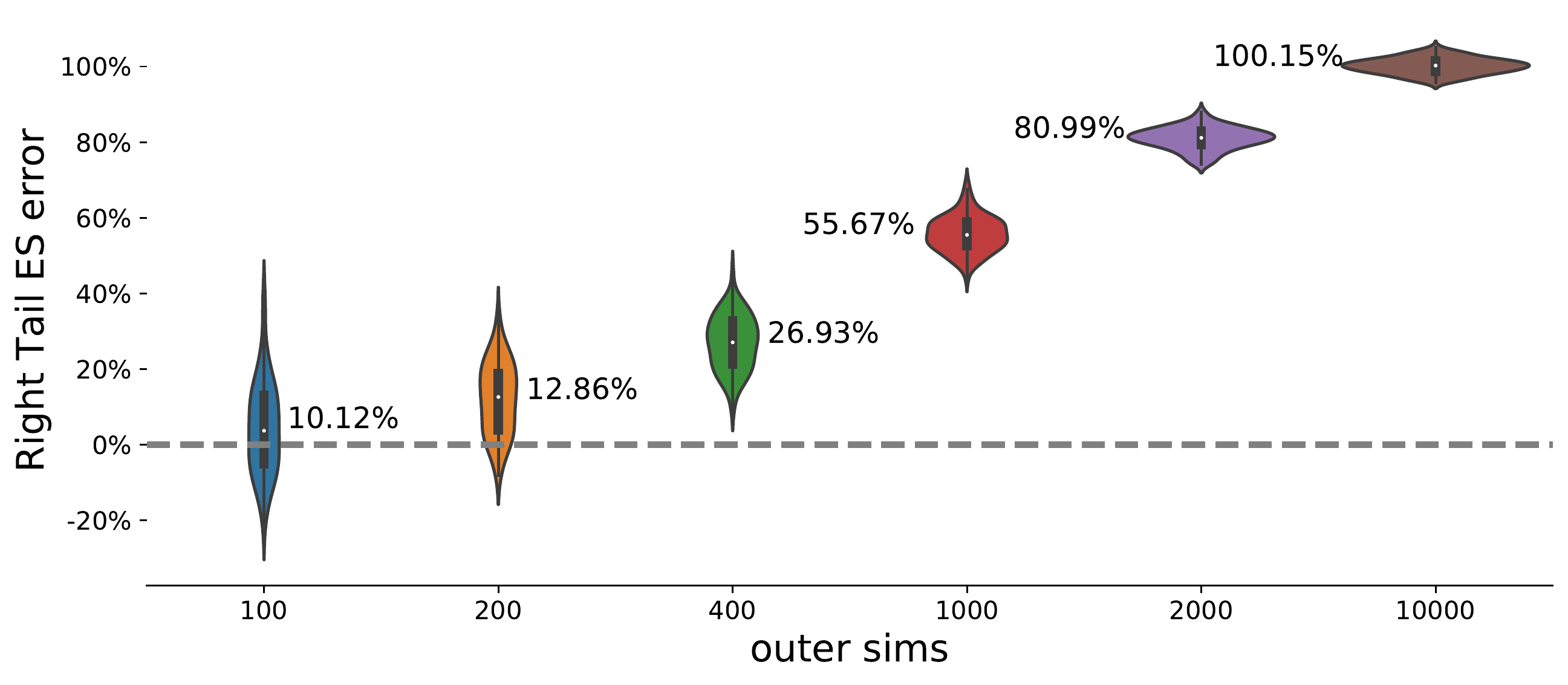}

    \caption{\small Error distribution of nMC estimators (fixed total simulations). Increasing outer simulation numbers lead to reduced variance but the corresponding decrease in inner simulation numbers increases bias.}\label{fig:nmc_results}
    
\end{figure}

Regarding the replicating portfolio benchmark, it is important to note that there is not \say{one} replicating portfolio method that we could directly apply. Replicating portfolios is a method with many variations and requires the choice of a range of hyper-parameters. Which explains the importance both of the description that follows and of our choices with regard to implementation. We have chosen parameters common in the industry except for one aspect. In the industry, replicating portfolio models are usually run many times with different parameters (in particular different asset classes) and one replicating portfolio is selected manually from those runs. Following that methodology would not allow us to create repeatable experiments or to form consistent distributions. We therefore use a Lasso regressor for the feature selection (instrument selection from the universe) and use the Bayesian information criterion to choose the regularization parameter. This is provided by the machine learning library scikit-learn through the LassoLarsIC class. As for the rest of the parameters, we list them in Table \ref{tab:rp}.

\begin{table}[]
\centering
\caption{Replicating portfolio parameters}\label{tab:rp}
\begin{tabular}{@{}ll@{}}
    \toprule
\textbf{Parameter}      & \textbf{Choice}\\
    \midrule
Loss function  & Squared errors                                                 \\
Asset universe & Bonds, cash, swaptions, equity index, equity European options \\
Time steps     & Full annual cash flow replicating (no grouping)                \\
Constraints    & None                                                           \\
Optimizer      & LassoLarsIC (scikit-learn)                              \\
    \bottomrule
\end{tabular}

\end{table}

\subsection{Quality measurement}
Following the arguments in \cite{willmott2005advantages} and \cite{chai2014root} we choose to focus on absolute errors rather than squared errors for model comparison since we do not need to penalize large outliers, and the errors are biased and most likely do not follow a normal distribution.  We therefore use mean absolute errors as a metric of mean model errors. Other moments are not quantitatively measured but the histograms of the distributions are presented for qualitative assessment.

The error to be considered is that of each risk metric (ES or VaR) and each tail (left tail or right tail), measured as a percentage error against the benchmark value for that metric. 

Based on the 100 macro-runs described in the previous section, we derive an empirical distribution for each estimator, $\{\rho_i\}_{i=1}^{100}$ where each individual sample is denoted as $\rho_i$. The mean absolute percentage error (MApE) is defined as
        \[ \frac{1}{100}\sum_i^R \big| \frac{\widehat{\rho_i}}{\rho} - 1 \big|, \] where $\rho$ is the benchmark value of the estimator.

All model results are mean-centred before risk metrics are calculated.

Figure \ref{fig:nmc_results} shows an example of the application of this error measurement for the various parameters in the nested Monte Carlo estimator. We can see that the estimator with 100 outer and 100 inner simulations has the lowest MApE. We therefore use this estimator in all subsequent comparison with other methods.

It is important to note that the benchmark value calculation is completely out-of-sample in respect of the training of any of the methods. Hence, all comparisons shown below are fully out-of-sample. 

\subsection{Results} \label{sec:results}
The results of the numerical experiments are presented in Table \ref{tab:results} and Figure \ref{fig:results}. The columns show the mean absolute percentage error for each of the four risk metrics and the mean present value. The rows show the errors of each different method. 

The errors of the replicating portfolio model are shown in the first row. In terms of mean absolute errors, this method yields the worst results of the group. Interestingly, the results are worse than using a nested Monte Carlo approach (second row). This is probably due to the very complex and non-linear nature of the guarantee function that is being replicated. Most likely the asset classes in the instrument universe are not complex enough, or not sufficiently path dependent to replicate the guarantee. 

While one might consider adding more asset classes, it becomes immediately clear that there is no obvious next step. This is a key disadvantage of replicating portfolios versus polynomial or neural network methods: the complexity of the approximation (richness of the approximating function) is not a parameter that one can easily modify. Which asset class to add next is an unknown, and there are hundreds of exotic derivatives that one might try. In order to test any of them, one must program the cash flows and valuation functions before any testing can be done. Therefore, the search for a better model takes much longer and success is not even guaranteed since one might not find the correct derivative.

By contrast, when working with a polynomial approximation one might increase the maximum degree of the polynomial, and with a neural network one might add layers or make each layer wider (adding more nodes). It only takes a few keystrokes to make a more powerful model. Polynomial and neural network models are guaranteed to succeed under certain conditions (smooth enough functions, sufficient samples, etc.), at least asymptotically. These models are, therefore, more amenable to automated solutions than are replicating portfolio models, which require an expert setup.

The last two rows in Table \ref{tab:results} present the results for each type of neural network models. Each model show better results than nested Monte Carlo or replicating portfolios. The first type (economic variable inputs) shows the best results of the group, better than the second type (random variable inputs).

\begin{table}
\centering
\small
\caption{Comparison of errors measured as MApEs} \label{tab:results}
    \begin{tabular}{@{}lrrrrr@{}}
    \toprule
    {} & \textbf{Left ES} & \textbf{Left VaR} &  \textbf{Mean} & \textbf{Right VaR} & \textbf{Right ES} \\
    \midrule
    \textbf{Rep. Portfolio MApE   } &     20\% &      23\% &   4\% &       46\% &      47\% \\
    \textbf{Nested MC MApE        } &     19\% &      14\% &   2\% &        8\% &      10\% \\
    \textbf{Neural net (econ) MApE} &      4\% &       2\% &   2\% &        7\% &       4\% \\
    \textbf{Neural net (rand) MApE} &     15\% &      11\% &   5\% &        4\% &       5\% \\
    \bottomrule
    \end{tabular}
\end{table}

Despite the advantages of the second type of neural network model, the data shows that 10,000 samples were not enough to learn the more complex function $CF \circ X$ sufficiently well to have higher accuracy than the first type of neural network model, which only had to learn the simpler function $CF$. In Figure \ref{fig:convergence} we can see that the second model (called \say{NNrand:10} when trained on 10,000 samples, \say{NNrand:50} when trained on 50,000 samples, etc.) does perform better than the first model (\say{NNecon}) once it is given a larger number of samples.
\begin{figure}
    \centering
    
    %\hspace*{-1.9cm}\includegraphics[scale=0.48]{rp_nn_errors.pdf}
    %\hspace*{-2cm}\includegraphics[scale=0.52]{rp_nn_nmc_errors_paper.pdf}
    \includegraphics[width=\columnwidth]{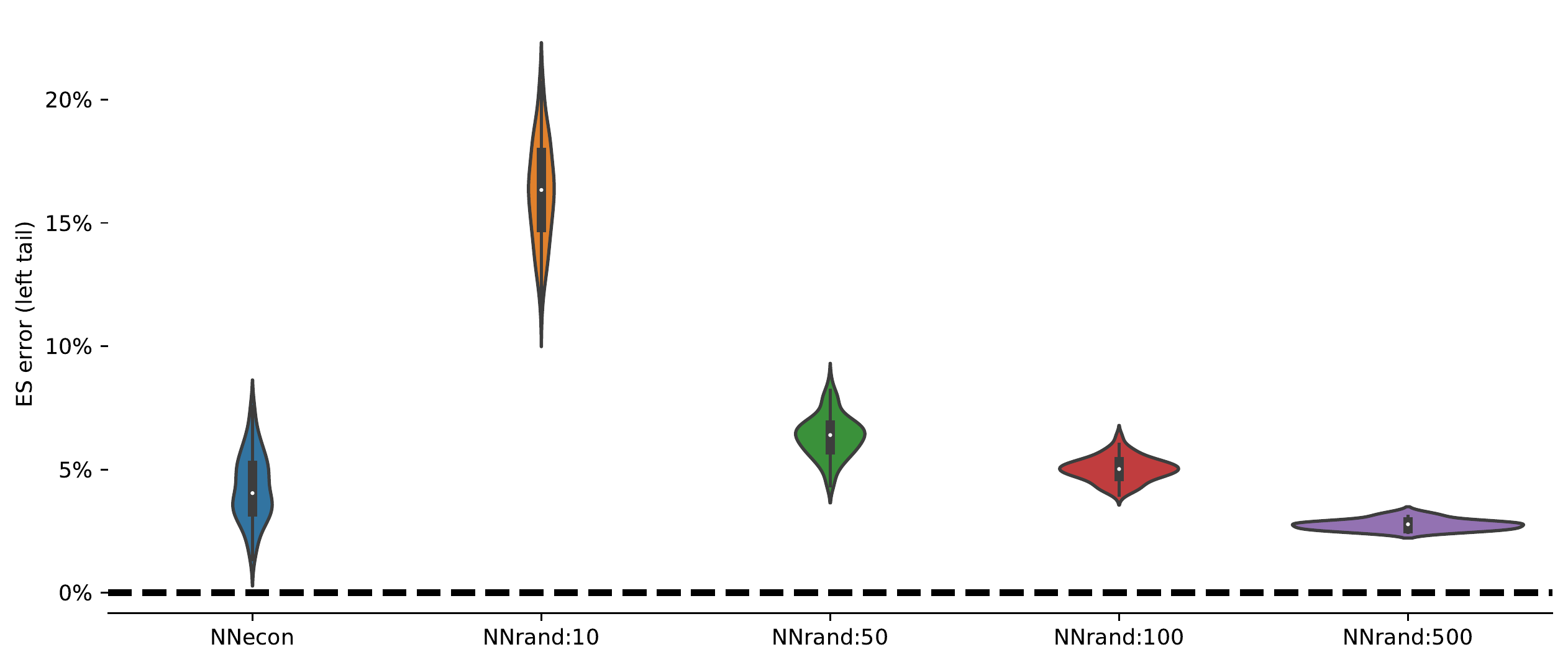}
    \caption{Effect of increasing training sample size on second neural network model. The decrease in mean and standard deviation of errors can be clearly seen as samples are increased from 10,000 to 500,000.}\label{fig:convergence}
    
\end{figure}

Figure \ref{fig:results} shows the empirical distribution of each of the estimators (trained, respectively, on 10,000 samples for comparability). We can observe the relative standard deviations of the estimators and find, as previously seen in Figure \ref{fig:nmc_results}, that nested Monte Carlo has low bias and a very high variance compared to the other methods. Interestingly, the standard deviations of the neural network models are not much higher than that of the replicating portfolios, despite each macro-run using a completely different set of inputs. This indicates that the calibration of the neural networks is robust to the variance of the inputs. In this regard, neural networks have a bad reputation due to their non-convex loss function. A common problem associated with the training of neural networks is that of local minima. Together with the common use of stochastic optimization methods (such as stochastic gradient descent), this usually contributes to a high variance of predictions. In this paper, we have taken certain steps to reduce this problem, by a) using a network of a small size (100 nodes) and b) using the Broyden–Fletcher–Goldfarb–Shanno (BFGS) optimization algorithm, which is non-stochastic. These decisions contribute to keeping the variance of the estimator within reasonable levels.

\begin{figure}[h]
    \centering
    
    %\hspace*{-1.9cm}\includegraphics[scale=0.48]{rp_nn_errors.pdf}
    %\hspace*{-2cm}\includegraphics[scale=0.52]{rp_nn_nmc_errors_paper.pdf}
    \includegraphics[width=\columnwidth]{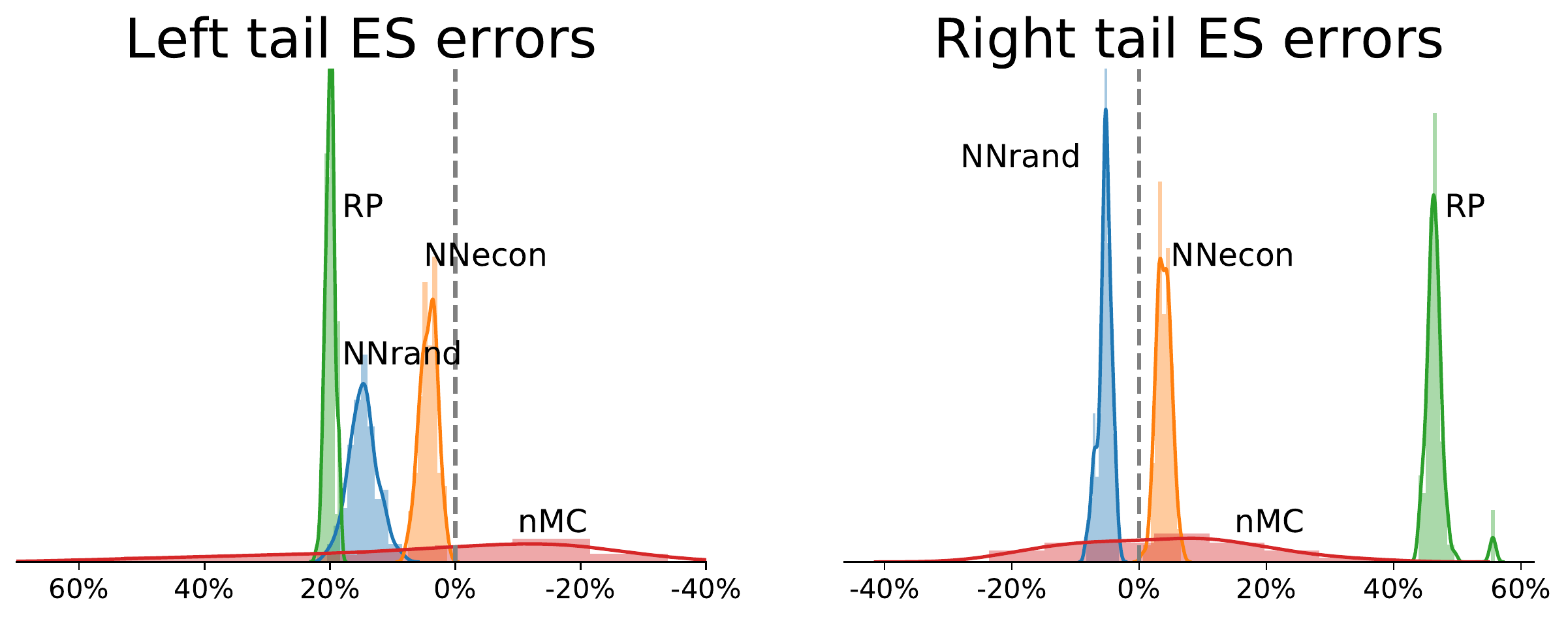}
    \caption{Empirical distribution of estimators (with 10,000 training samples). Each neural network model delivers more accurate results than the benchmark models.}\label{fig:results}
    
\end{figure}

\subsection{Runtime}
As shown in Table \ref{tab:results} and Figure \ref{fig:results}, the neural network models perform very well in terms of quality. Since these models require a non-linear regression to find their parameters, they can be expected to be slower than a purely linear model, such as a replicating portfolio. Table \ref{tab:runtime} shows how long it takes to run each model \say{end to end}---that is, including feature generation, calibration and prediction.

\begin{table}[h]
\centering
\small
\caption{Comparison of runtime (in seconds) for training data sets of different sizes} \label{tab:runtime}

    \begin{tabular}{@{}rrrr@{}}
    \toprule
    \textbf{Samples} & \textbf{Neural net (econ)} & \textbf{Neural net (rand)} & \textbf{Rep. portfolio} \\
    \midrule
    \textbf{2500   } &     112 &     115 &  4 \\
    \textbf{5000   } &     202 &     200 &  4 \\
    \textbf{10000  } &     384 &     406 &  6 \\
    \textbf{20000  } &     707 &     711 & 13 \\
    \textbf{40000  } &    1458 &    1407 & 26 \\
    \bottomrule
    \end{tabular}

\end{table}

We can see that training a neural network model takes longer than training a replicating portfolio. Both types of model scale approximately linearly to the size of the training set. However, even at the far end (40,000 training samples) the neural network models do not take more than 30 minutes (on a single-core AMD Opteron 6380) which is perfectly acceptable from a practitioner's point of view. Generating the inputs required for economic capital calculations can normally take from several hours to several days; and 30 minutes can therefore fit easily within the normal production schedule of an insurance company.

\section{Conclusions}
Based on a simulated insurance product, we have presented a comparison of nested Monte Carlo, replicating portfolios, and neural networks as methods for calculating solvency capital. The numerical experiments show that neural networks perform very well, even in a highly non-linear problem with a small number of training samples. The mean errors are the lowest of the group and the distributions of the results do not show qualitatively, large variance. A qualitative analysis suggests that the neural network model can also have advantages in terms of model simplicity.

In the construction of this neural network model we make use of what we believe is a novel formula for calculating the risk-neutral price of a neural network.
Additionally, we make two contributions for other researchers in the field: a description of an automated replicating portfolio model (necessary for reliable comparisons) and a full data set and software library for the production of economic scenarios.

\section*{Acknowledgements}
This paper is based on our presentation at the 2019 Insurance Data Science Conference (ETH Zurich). We thank the scientific committee for its invitation to present, and all those that provided helpful comments during the preparation: Mariana Andres, Dr Gregor Reich, and the participants of the Quantitative Business Administration PhD Seminar. Special thanks go to Prof. Karl Schmedders and Prof. Damir Filipovic, for reviewing the presentation and for all we have learnt from them, without which this paper would not be possible.

\clearpage
\bibliographystyle{plain}
\bibliography{bibtex}

\begin{thebibliography}{}

\bibitem[\protect\citeauthoryear{Adelmann, Fernandez~Arjona, Mayer, and
  Schmedders}{Adelmann et~al.}{2019}]{fernandez2016large}
Adelmann, M., L.~Fernandez~Arjona, J.~Mayer, and K.~Schmedders (2019).
\newblock A large-scale optimization model for replicating portfolios in the
  life insurance industry.
\newblock {\em Working Paper, University of Zurich\/}.

\bibitem[\protect\citeauthoryear{Bauer, Bergmann, and Reuss}{Bauer
  et~al.}{2010}]{bauer2010solvency}
Bauer, D., D.~Bergmann, and A.~Reuss (2010).
\newblock {Solvency II and nested simulations--a least-squares Monte Carlo
  approach}.
\newblock In {\em Proceedings of the 2010 ICA congress}.

\bibitem[\protect\citeauthoryear{Beutner, Pelsser, and Schweizer}{Beutner
  et~al.}{2013}]{beutner2013fast}
Beutner, E., A.~Pelsser, and J.~Schweizer (2013).
\newblock {Fast convergence of regress-later estimates in least squares Monte
  Carlo}.
\newblock {\em Available at SSRN 2328709\/}.

\bibitem[\protect\citeauthoryear{Beutner, Pelsser, and Schweizer}{Beutner
  et~al.}{2016}]{beutner2016theory}
Beutner, E., A.~Pelsser, and J.~Schweizer (2016).
\newblock Theory and validation of replicating portfolios in insurance risk
  management.
\newblock {\em Available at SSRN 2557368\/}.

\bibitem[\protect\citeauthoryear{Broadie, Du, and Moallemi}{Broadie
  et~al.}{2015}]{broadie2015risk}
Broadie, M., Y.~Du, and C.~C. Moallemi (2015).
\newblock Risk estimation via regression.
\newblock {\em Operations Research\/}~{\em 63\/}(5), 1077--1097.

\bibitem[\protect\citeauthoryear{Cambou and Filipovi{\'c}}{Cambou and
  Filipovi{\'c}}{2018}]{cambou2018replicating}
Cambou, M. and D.~Filipovi{\'c} (2018).
\newblock Replicating portfolio approach to capital calculation.
\newblock {\em Finance and Stochastics\/}~{\em 22\/}(1), 181--203.

\bibitem[\protect\citeauthoryear{Castellani, Fiore, Marino, Passalacqua, Perla,
  Scognamiglio, and Zanetti}{Castellani
  et~al.}{2018}]{castellani2018investigation}
Castellani, G., U.~Fiore, Z.~Marino, L.~Passalacqua, F.~Perla, S.~Scognamiglio,
  and P.~Zanetti (2018).
\newblock An investigation of machine learning approaches in the solvency ii
  valuation framework.
\newblock {\em Available at SSRN 3303296\/}.

\bibitem[\protect\citeauthoryear{Chai and Draxler}{Chai and
  Draxler}{2014}]{chai2014root}
Chai, T. and R.~R. Draxler (2014).
\newblock Root mean square error (rmse) or mean absolute error
  (mae)?--arguments against avoiding rmse in the literature.
\newblock {\em Geoscientific model development\/}~{\em 7\/}(3), 1247--1250.

\bibitem[\protect\citeauthoryear{Chen and Skoglund}{Chen and
  Skoglund}{2012}]{chen2012cashflow}
Chen, W. and J.~Skoglund (2012).
\newblock Cashflow replication with mismatch constraints.
\newblock {\em The Journal of Risk\/}~{\em 14\/}(4), 115.

\bibitem[\protect\citeauthoryear{Fernandez-Arjona}{Fernandez-Arjona}{2019}]{dataset}
Fernandez-Arjona, L. (2019).
\newblock Esg simulations and rpdb cash flows (june 2019).
\newblock \url{http://dx.doi.org/10.17632/6vvzh2w4g4.1}.
\newblock Mendeley Data, v1.

\bibitem[\protect\citeauthoryear{Gan and Lin}{Gan and
  Lin}{2015}]{gan2015valuation}
Gan, G. and X.~S. Lin (2015).
\newblock Valuation of large variable annuity portfolios under nested
  simulation: A functional data approach.
\newblock {\em Insurance: Mathematics and Economics\/}~{\em 62}, 138--150.

\bibitem[\protect\citeauthoryear{Glasserman}{Glasserman}{2013}]{glasserman2013monte}
Glasserman, P. (2013).
\newblock {\em Monte Carlo methods in financial engineering}, Volume~53.
\newblock Springer Science \& Business Media.

\bibitem[\protect\citeauthoryear{Glasserman and Yu}{Glasserman and
  Yu}{2002}]{glasserman2002simulation}
Glasserman, P. and B.~Yu (2002).
\newblock Simulation for american options: Regression now or regression later?
\newblock In {\em Monte Carlo and Quasi-Monte Carlo Methods 2002}, pp.\
  213--226. Springer.

\bibitem[\protect\citeauthoryear{Hanin and Sellke}{Hanin and
  Sellke}{2017}]{hanin2017approximating}
Hanin, B. and M.~Sellke (2017).
\newblock Approximating continuous functions by relu nets of minimal width.

\bibitem[\protect\citeauthoryear{Hejazi and Jackson}{Hejazi and
  Jackson}{2016}]{hejazi2016neural}
Hejazi, S.~A. and K.~R. Jackson (2016).
\newblock A neural network approach to efficient valuation of large portfolios
  of variable annuities.
\newblock {\em Insurance: Mathematics and Economics\/}~{\em 70}, 169--181.

\bibitem[\protect\citeauthoryear{Hornik}{Hornik}{1991}]{hornik1991approximation}
Hornik, K. (1991).
\newblock Approximation capabilities of multilayer feedforward networks.
\newblock {\em Neural networks\/}~{\em 4\/}(2), 251--257.

\bibitem[\protect\citeauthoryear{Lee and Carter}{Lee and
  Carter}{1992}]{leecarter}
Lee, R.~D. and L.~R. Carter (1992).
\newblock Modeling and forecasting u. s. mortality.
\newblock {\em Journal of the American Statistical Association\/}~{\em
  87\/}(419), 659--671.

\bibitem[\protect\citeauthoryear{Longstaff and Schwartz}{Longstaff and
  Schwartz}{2001}]{LongstaffSchwartz}
Longstaff, F.~A. and E.~S. Schwartz (2001).
\newblock Valuing american options by simulation: A simple least-squares
  approach.
\newblock {\em The Review of Financial Studies\/}~{\em 14\/}(1), 113--147.

\bibitem[\protect\citeauthoryear{McKinney}{McKinney}{2010}]{mckinney-proc-scipy-2010}
McKinney, W. (2010).
\newblock Data structures for statistical computing in python.
\newblock In S.~van~der Walt and J.~Millman (Eds.), {\em Proceedings of the 9th
  Python in Science Conference}, pp.\  51 -- 56.

\bibitem[\protect\citeauthoryear{Natolski and Werner}{Natolski and
  Werner}{2014}]{natolski2014mathematical}
Natolski, J. and R.~Werner (2014).
\newblock Mathematical analysis of different approaches for replicating
  portfolios.
\newblock {\em European Actuarial Journal\/}~{\em 4\/}(2), 411--435.

\bibitem[\protect\citeauthoryear{Oliphant}{Oliphant}{06  }]{numpy}
Oliphant, T. (2006--).
\newblock {NumPy}: A guide to {NumPy}.
\newblock USA: Trelgol Publishing.
\newblock [Online; accessed November 2019].

\bibitem[\protect\citeauthoryear{Pedregosa, Varoquaux, Gramfort, Michel,
  Thirion, Grisel, Blondel, Prettenhofer, Weiss, Dubourg, et~al.}{Pedregosa
  et~al.}{2011}]{pedregosa2011scikit}
Pedregosa, F., G.~Varoquaux, A.~Gramfort, V.~Michel, B.~Thirion, O.~Grisel,
  M.~Blondel, P.~Prettenhofer, R.~Weiss, V.~Dubourg, et~al. (2011).
\newblock Scikit-learn: Machine learning in python.
\newblock {\em Journal of machine learning research\/}~{\em 12\/}(Oct),
  2825--2830.

\bibitem[\protect\citeauthoryear{Pelsser and Schweizer}{Pelsser and
  Schweizer}{2016}]{pelsser2016difference}
Pelsser, A. and J.~Schweizer (2016).
\newblock The difference between lsmc and replicating portfolio in insurance
  liability modeling.
\newblock {\em European actuarial journal\/}~{\em 6\/}(2), 441--494.

\bibitem[\protect\citeauthoryear{Tibshirani}{Tibshirani}{1996}]{tibshirani1996regression}
Tibshirani, R. (1996).
\newblock Regression shrinkage and selection via the lasso.
\newblock {\em Journal of the Royal Statistical Society: Series B
  (Methodological)\/}~{\em 58\/}(1), 267--288.

\bibitem[\protect\citeauthoryear{Van Der~Walt, Colbert, and Varoquaux}{Van
  Der~Walt et~al.}{2011}]{van2011numpy}
Van Der~Walt, S., S.~C. Colbert, and G.~Varoquaux (2011).
\newblock The numpy array: a structure for efficient numerical computation.
\newblock {\em Computing in Science \& Engineering\/}~{\em 13\/}(2), 22.

\bibitem[\protect\citeauthoryear{Vidal and Daul}{Vidal and
  Daul}{2009}]{vidal2009replication}
Vidal, E.~G. and S.~Daul (2009).
\newblock Replication of insurance liabilities.
\newblock {\em RiskMetrics Journal\/}~{\em 9\/}(1).

\bibitem[\protect\citeauthoryear{Willmott and Matsuura}{Willmott and
  Matsuura}{2005}]{willmott2005advantages}
Willmott, C.~J. and K.~Matsuura (2005).
\newblock Advantages of the mean absolute error (mae) over the root mean square
  error (rmse) in assessing average model performance.
\newblock {\em Climate research\/}~{\em 30\/}(1), 79--82.

\bibitem[\protect\citeauthoryear{Zou, Hastie, Tibshirani, et~al.}{Zou
  et~al.}{2007}]{zou2007degrees}
Zou, H., T.~Hastie, R.~Tibshirani, et~al. (2007).
\newblock On the “degrees of freedom” of the lasso.
\newblock {\em The Annals of Statistics\/}~{\em 35\/}(5), 2173--2192.

\end{thebibliography}

\end{document}